\documentclass[journal]{IEEEtran}
\usepackage{theorem} \usepackage{cite} \usepackage{stfloats} \usepackage{epsfig} \usepackage{verbatim}
\usepackage{graphicx} \usepackage{amssymb} \usepackage{amsmath} \usepackage{color}
\usepackage[normalem]{ulem}
\usepackage{bm}
\usepackage{xcolor}
\usepackage{algorithm}
\usepackage{algorithmicx}
\usepackage{algpseudocode}
\usepackage{amsfonts}
\usepackage{cases}
\usepackage{textcomp}
\usepackage{epstopdf}
\DeclareMathOperator*{\argmax}{argmax}
\ifCLASSOPTIONcompsoc
	\usepackage[caption=false,font=normalsize,labelfont=sf,textfont=sf]{subfig}
\else
	\usepackage[caption=false,font=footnotesize]{subfig}
\fi
\def\BibTeX{{\rm B\kern-.05em{\sc i\kern-.025em b}\kern-.08em
    T\kern-.1667em\lower.7ex\hbox{E}\kern-.125emX}}

\newtheorem{theorem}{Theorem}

\newenvironment{proof}{{\indent \indent \it Proof:}}{\hfill $\square$\par}

\begin{document}

\title{A 3D Positioning-based Channel Estimation Method 
for RIS-aided mmWave Communications}

\author{Yaoshen~Cui, 
        Haifan~Yin, 
         Li~Tan,
         and~Marco~Di~Renzo,~\IEEEmembership{Fellow,~IEEE}
\thanks{This work was supported by the National Natural Science Foundation of China under grant No. 62071191. The corresponding author is Li Tan.}
\thanks{Yaoshen Cui, Haifan Yin, and Li Tan are with the School of Electronic Information and Communications, 
 Huazhong University of Science and Technology, 
 Wuhan, 430074, China (email: yaoshen\_cui@hust.edu.cn;  yin@hust.edu.cn; ltan@hust.edu.cn). 
}
\thanks{M. Di Renzo is with Universit\'e Paris-Saclay, CNRS, CentraleSup\'elec, Laboratoire des Signaux et Syst\`emes, 3 Rue Joliot-Curie, 91192 Gif-sur-Yvette, France. (marco.di-renzo@universite-paris-saclay.fr)} 
\thanks{A part of this work \cite{Cui:2021ICCC}
was presented in the 10th IEEE/CIC International Conference on Communications in China (IEEE/CIC ICCC 2021).}
}

\maketitle

\begin{abstract}
A fundamental challenge in millimeter-wave (mmWave) communication is the susceptibility to blocking objects. One way to alleviate this problem is the use of reconfigurable intelligent surfaces (RIS). Nevertheless, due to the large number of passive reflecting elements on RIS, channel estimation turns out to be a challenging task. In this paper, we address the channel estimation for RIS-aided mmWave communication systems based on a localization method. The proposed idea consists of exploiting the sparsity of the mmWave channel and the topology of the RIS. In particular, we first propose the concept of reflecting unit set (RUS) to improve the flexibility of RIS. We then propose a novel coplanar maximum likelihood-based (CML) 3D positioning method based on the RUS, and derive the Cramer-Rao lower bound (CRLB) for the positioning method. Furthermore, we develop an efficient positioning-based channel estimation scheme with low computational complexity. Compared to state-of-the-art methods, our proposed method requires less time-frequency resources in channel acquisition as the complexity is independent to the total size of the RIS but depends on the size of the RUSs, which is only a small portion of the RIS. Large performance gains are confirmed in simulations, which proves the effectiveness of the proposed method.
\end{abstract}


\IEEEpeerreviewmaketitle

\section{Introduction}

The commercialization of the fifth-generation (5G)  of communications is in full swing and is facilitated by several key technologies, including massive multiple-input multiple-output (MIMO) and millimeter-wave (mmWave) communications. These technologies leads to significantly increase of the spectral efficiency and alleviate the shortage of spectrum in wireless communication networks. 
However, mmWave communication systems are much more susceptible to blockage effects compared with lower frequency band communication systems\cite{TS:2013MmwaveWork}, which may undermine the potential of mmWave communications in some network deployments.
This problem at mmWave frequencies calls for solutions to make full use of the large available bandwidth.

A possible solution to counteract the coverage problem of mmWave systems is the deployment of  reconfigurable intelligent surfaces (RIS) \cite{Marco:2020Reconfigurable,Marco:2020Smart,Marco:2021communication}(also known as intelligent reflecting surfaces (IRSs)\cite{wuqq:2019IRS} and metasurfaces\cite{Tang:2019meta}).
RIS is an emerging technology with the capability of intelligently manipulating the electromagnetic waves, which has lately drawn much attention from academia and industry as it has the potential of controlling the wireless propagation environment at a low cost and energy consumption.
The authors of \cite{Tang:2020wireless,tang2021path} developed free-space path loss channel models for RIS-aided systems based on the electromagnetic characteristics of the RIS and validated them through simulations.
In \cite{Dai:2020RIS}, the authors developed a new type of high-gain and low-cost RIS
with 256 reflecting elements, which operates at 2.3 GHz and 28.5 GHz. The authors of \cite{Pei:2021TCOM} demonstrated a prototype of RIS-aided communication system and showed some indoor and outdoor field trials with a fabricated RIS endowed with 1100 reflecting elements. 
The authors of \cite{hou_reconfigurable_2020} conceived a system for serving paired power-domain non-orthogonal multiple access (NOMA) users by designing the passive beamforming coefficients at the RIS.

In theory, the full potential of RIS may be achieved when the channel state information (CSI) is known. However, the RIS is usually not equipped with power amplifiers, digital signal processing units, multiple radio frequency chains and the number of the reflecting elements is usually large, thus making the CSI estimation a challenging problem.

In literature, several methods have been proposed to solve estimate the channel information in RIS-aided systems\cite{pan2021overview}. These include the cascaded channel decomposition approach \cite{mishra_channel_2019,jensen_optimal_2020,wang_channel_2020,wei_channel_2021}, and methods based on the sparsity of channels\cite{chen_channel_2019,wang_compressed_2020,He:2020CascadedChannel}, the quasi-static property \cite{Liu:2020MatrixCalibration,hu_two-timescale_2021}, etc.
In a multiple-input single-output (MISO) system, the authors of  \cite{mishra_channel_2019,jensen_optimal_2020} proposed to decompose the cascaded RIS channel to sub-channels, each of them corresponding to an RIS element, which are easier to be estimated.
In a multi-user system, the authors of \cite{wang_channel_2020} proposed a three-phase pilot-based channel estimation protocol by activating users one by one. The authors of \cite{wei_channel_2021} developed two iterative estimation algorithms based on the parallel factor (PARAFAC) decomposition method.
By exploiting the angular sparsity of the channels, the authors of \cite{chen_channel_2019,wang_compressed_2020} formulated the cascaded CSI estimation  as a sparse signal recovery problem, which is solved using compressive sensing (CS) algorithms.
The authors of \cite{He:2020CascadedChannel} proposed a two-stage channel estimation algorithm by  bilinear sparse matrix factorization method.
Since the positions of the base station (BS) and the RIS are usually fixed, the BS-RIS channel varies slowly and therefore can be assumed quasi-static. The paper \cite{Liu:2020MatrixCalibration,hu_two-timescale_2021} developed a matrix-calibration-based method and a
two-timescale channel estimation framework, respectively. These methods require less training overhead by exploiting the quasi-static property of the BS-RIS channel.
However, these traditional channel estimation methods generally require huge time-frequency resources or high computational complexity, especially when the number of the reflecting elements is very large.

RIS has several applications for positioning as well\cite{Abrardo:2021Intelligent}.
In \cite{Henk:2020beyond5G}, the authors studied the possibility of RIS-aided localization and derived the corresponding Fisher information matrix (FIM) and the CRLB in closed-form expressions.
The authors of \cite{Ma:2020indoor} utilized an RIS to facilitate the positioning in conjunction with an ultra-wideband (UWB) indoor systems.
The authors of \cite{He:2020adaptive} considered an RIS-aided mmWave MIMO system when the direct line of sight (LOS) is obstructed, and proposed an adaptive phase adjustment method based on hierarchical codebooks and feedbacks from mobile stations to enable joint localization and communication.
Furthermore, conventional localization methods, e.g., those based on the received signal strength (RSS), the time of arrival (TOA), the time difference of arrival (TDOA) and the angle of arrival (AOA), can also be utilized for improving the localization performance in mmWave and massive MIMO systems\cite{Lemic:2016localization}.

In this paper, we propose a low-complexity channel estimation scheme that is based on a novel localization method. The method is designed for application to RIS-aided mmWave communication systems, and it exploits 1) the geometry of the RIS, 2) the structural information of the wideband channels, and  3) a new low-complexity coplanar maximum likelihood-based 3D positioning method.
The main contributions of this paper are summarized below:
\begin{itemize}
\item
    We introduce the concept of \emph{Reflecting Unit Set} (RUS), which is a set of neighboring reflecting units in a specific part of the RIS. A joint minimum mean square estimator (JMMSE) is proposed to calculate the distances between the RUSs and the user equipment (UE) based on the low-rank property of the wideband channel covariance matrix in the frequency domain. 
    We utilize the RUSs as anchors to compute the 3D position of the UE by exploiting the spatial structure of the wideband channel. 
\item
     We propose a coplanar maximum likelihood-based (CML) three-dimensional localization scheme that offers a high accuracy. Since the reflecting units of the RIS are coplanar, determining the 3D position of a target is a challenging problem. Closed-form expressions are derived for the localization of the target with the help of the RUSs based on the maximum likelihood (ML) criterion and the principles of Euclidean geometry.  
\item
    We derive a closed-form expression of the CRLB for the proposed scheme. We also show that its complexity does not depend on the number of reflecting units on the RIS, but it depends only on the number and size of the RUSs. As the number of the reflecting elements increases, the complexity of our proposed scheme remains low. 
\end{itemize}

With the aid of numerical results, we show that the proposed method requires much lower computational complexity and uses much fewer time-frequency resources for CSI acquisition compared with the state-of-the-art methods\cite{mishra_channel_2019,jensen_optimal_2020,wang_channel_2020,wei_channel_2021,chen_channel_2019,wang_compressed_2020,He:2020CascadedChannel,Liu:2020MatrixCalibration,hu_two-timescale_2021}. The complexity and the pilot overhead of the state-of-the-art methods depend on the total size of the RIS, which grows rapidly when increasing the size of the RIS. In contrast, the required complexity and the time-frequency resources of the proposed method are just  related to the size of a small portion of the RIS.
Specially, we show that the proposed method needs around 20 orthogonal frequency division multiplexing (OFDM) symbols for channel estimation (or 0.167 ms with 120 kHz of subcarrier spacing in mmWave systems) \cite{3gpp:38.211}, even with thousands reflecting units on the RIS. To the best of our knowledge, this is the first low-complexity channel estimation method for RIS-aided systems based on the RUS concept.

The remainder of this paper is organized as follows. In Section \ref{sec:model}, we describe the considered RIS-aided mmWave system and channel model. In Section \ref{sec:proposed}, we propose the channel estimation method based on 3D localization. In Section \ref{sec:simulation}, we present numerical results for the proposed method. Finally, Section \ref{sec:conclusion} concludes the paper.

\emph{Notation:} Bold-face symbols denote vectors and matrices. $\mathbb{C}^{x\times y}$ is the space of $x\times y$ complex-valued matrices. $\odot$, $\lfloor \cdot \rfloor$  and $\mod$ stand for the Hadamard product, the rounding down and the remainder operation, respectively.
For a vector $\bm{x}$, $\|\bm{x}\|$ is its Euclidean norm and $\text{diag}(\bm{x})$ denotes a diagonal matrix with each diagonal element being the corresponding element in $\bm{x}$. $\bm{X}^T$, $\bm{X}^*$, $\bm{X}^H$ and $\text{tr}(\bm{X})$ are the transpose, conjugate, conjugate transpose and trace of a matrix $\bm{X}$, respectively. $\bm{I}$ denotes the identity matrix.

\section{RIS-aided Wireless Communications and Channel modeling }\label{sec:model}

\begin{figure}[!t]
\centering
\includegraphics[width=3.4in]{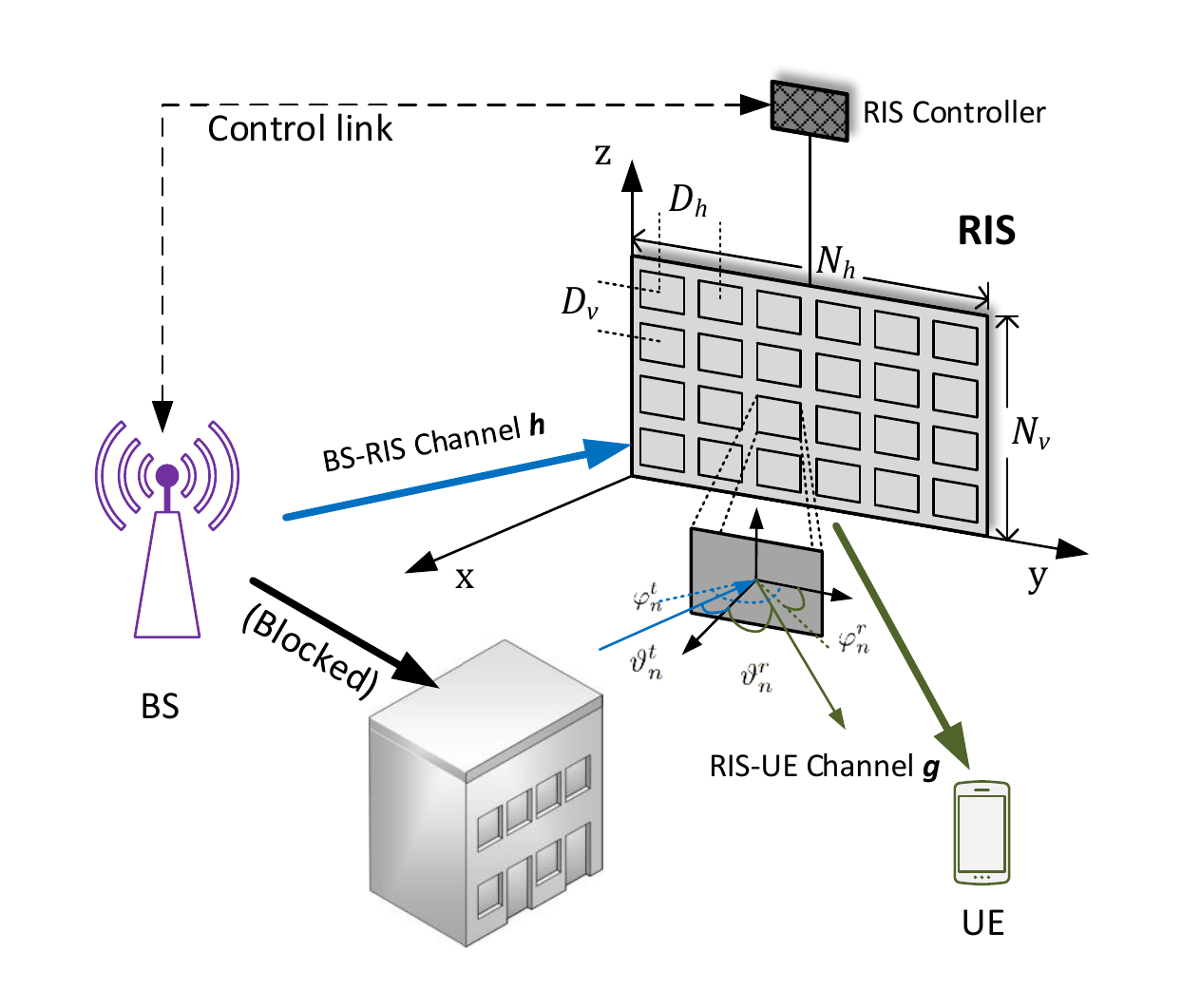}
\caption{RIS-assisted wireless communication system}
\label{fig:RIS-assisted_Communication_System}
\end{figure}

Consider an RIS-aided communication system where the LOS path between the BS and the UE is blocked, as shown in Fig. \ref{fig:RIS-assisted_Communication_System}.
Due to the high transmission frequency of mmWave, the number of paths is much smaller than in sub-6 GHz bands.
The RIS contains $N$ reflecting elements and is deployed on the y-z plane of a Cartesian coordinate system. Without loss of generality, we assume the bottom left corner of the RIS coincides with the origin of the coordinate system.
The RIS is modeled as a uniform planar array (UPA) with $N_h$ columns and $N_v$ rows. The area of each element is $A_u=D_h\times D_v$, where $D_h$ and $D_v$ are the element spacing in the horizontal and vertical directions, respectively.
For each reflecting element, the gain is denoted by $G_u$, and the 
normalized power radiation pattern\footnote{The power radiation pattern quantifies the amount of power radiated or received by an antenna as a function of the direction of observation with respect to the center of the antenna.}  is $F(\vartheta,\varphi)$, with $\vartheta$ and $\varphi$ being the zenith and azimuth angles, respectively.

$F_n^t=F(\vartheta_n^{t},\varphi_n^{t})$ and $F_n^r=F(\vartheta_n^{r},\varphi_n^{r})$ are the normalized power radiation pattern from the $n$-th reflecting unit of the RIS to the BS and from the $n$-th reflecting unit of the RIS to the UE, respectively, where $\vartheta_n^{t}$, $\varphi_n^{t}$, $\vartheta_n^{r}$ and  $\varphi_n^{r}$ are the corresponding zenith and azimuth angles.

We introduce the vectors $\bm{\phi}=[\phi_{1},\ldots,\phi_{\emph{N}}]$ and $\bm{\beta}=[\beta_{1},\ldots,\beta_{\emph{N}}]$, where $\phi_{n}\in[0,2\pi)$ and $\beta_{n}\in[0,1], (n = 1,\ldots,N)$  denote the phase shift and the amplitude of the reflection coefficient of the $n$-th reflecting unit of the RIS, respectively.
The vector of reflection coefficients of the RIS,  denoted by $\bm{\theta}\in \mathbb{C}^{N\times 1}$, is expressed as
\begin{equation}\label{eq:ref-coef vector of the RIS}
\begin{split}
\bm{\theta}&={[\theta_{1},\cdots,\theta_{\emph{N}}]}^T\\
           &={[\beta_{1}e^{j\phi_{1}},\ldots,\beta_{\emph{N}}e^{j\phi_{\emph{N}}}]}^T,
\end{split}
\end{equation}
where $\theta_{n}=\beta_{n}e^{j\phi_{n}}(n=1,\ldots,N)$ is the reflection coefficient of the $n$-th reflection unit of the RIS.

We mainly consider the LOS links between the transceivers and the RIS, which are usually much stronger than other multipath components in the mmWave band \cite{Di:practical2020,Basar:simris2020}.
Without loss of generality, we assume the BS and the UE both have single antenna, yet the proposed method can be readily generalized to multiple-antenna BS/UE cases.
The channel between the BS and the $n$-th reflecting unit is modeled as\cite{Tang:2020wireless}
\begin{equation}\label{eq:BS-n channel}
 h_n=\sqrt{\frac{A_uF^{tx}_nF^t_n}{4\pi d_{t,n}^2}} e^{\frac{-j2\pi d_{t,n}}{\lambda}},
\end{equation}
where $\lambda$ is the wavelength, $F^{tx}_n = F^{tx}(\vartheta_n^{tx},\varphi_n^{tx})$ is the radiation pattern of the transmit antenna from the BS to the $n$-th reflecting element of the RIS with $\vartheta_n^{tx}$ and $\varphi_n^{tx}$ denoting the corresponding zenith and azimuth angles, and $d_{t,n}$ denoting the distance between the BS and the $n$-th reflecting element of the RIS.

The channel vector between the BS and the RIS, which is denoted by $\bm{h}\in \mathbb{C}^{N\times 1}$,  is thus expressed as
\begin{equation}\label{eq:BS-RIS channel}
\begin{split}
\bm{h}&=[h_{1},\cdots,h_{\emph{N}}]^T\\
      &=\resizebox{.8\columnwidth}{!}{$\bigg[\sqrt{\frac{A_uF^{tx}_1F^t_1}{4\pi d_{t,1}^2}} e^{\frac{-j2\pi d_{t,1}}{\lambda}}, \cdots, \sqrt{\frac{A_uF^{tx}_NF^t_N}{4\pi d_{t,N}^2}}e^{\frac{-j2\pi d_{t,N}}{\lambda}} \bigg]^T$}.
\end{split}
\end{equation}

Similarly, the channel between the UE and the $n$-th reflecting element is expressed as
\begin{equation}\label{eq:UE-n channel}
 g_n=\sqrt{\frac{A_rF^{rx}_nF_n^r}{4\pi d_{r,n}^2}} e^{\frac{-j2\pi d_{r,n}}{\lambda}},
\end{equation}
where $A_r$ 
is the aperture of the receiving antenna, $F^{rx}_n = F^{rx}(\vartheta_n^{rx},\varphi_n^{rx})$ is the normalized radiation pattern of the receive antenna from the UE to the $n$-th reflecting element of the RIS with $\vartheta_n^{rx}$ and $\varphi_n^{rx}$ being the zenith and azimuth angles, and $d_{r,n}$ being the distance between the UE and the $n$-th reflecting element of the RIS.

Based on (\ref{eq:UE-n channel}), the LOS channel vector between the UE and all the elements on the RIS is expressed as:
\begin{equation}\label{eq:UE-RIS channel}
\begin{split}
\bm{g}&=[g_{1},\cdots,g_{\emph{N}}]^T\\
      &=\resizebox{.8\columnwidth}{!}{$\bigg[\sqrt{\frac{A_rF^{rx}_1F^r_1}{4\pi d_{r,1}^2}} e^{\frac{-j2\pi d_{r,1}}{\lambda}}, \cdots, \sqrt{\frac{A_rF^{rx}_NF^r_N}{4\pi d_{r,N}^2}}e^{\frac{-j2\pi d_{r,N}}{\lambda}} \bigg]^T$}.
\end{split}
\end{equation}

Given the BS-RIS channel and the RIS-UE channel in (\ref{eq:BS-RIS channel}) and (\ref{eq:UE-RIS channel}), the cascaded BS-RIS-UE channels are expressed as:
\begin{equation}\label{eq:BS-RIS-UE channel}
\begin{split}
   w & = \left((\sqrt{G_r}\bm{g})\odot(\sqrt{G_t}\bm{h})\right)^T(\sqrt{G_u}\bm{\theta})\\
     & = \sqrt{G_rG_uG_t}(\bm{g}\odot\bm{h})^T\bm{\theta}\\
     & =\frac{\sqrt{G_rG_uG_tA_rA_u}}{4\pi}\sum_{n=1}^{N}\frac{\sqrt{\tilde{F}_n}\theta_n}{d_{t,n}d_{r,n}}   e^{\frac{-j2\pi(d_{t,n}+d_{r,n})}{\lambda}}
\end{split}
\end{equation}
where $G_t$ denotes the transmit antenna gain, $G_r$ denotes the receive antenna gain and $\tilde{F}_n=F^{tx}_nF^t_nF^{rx}_nF^r_n$ accounts for the effect of the normalized power radiation patterns on the received signal power.

The signal received by the UE is
\begin{equation}\label{eq:y-signal recevied}
\begin{split}
r &= sw +\eta\\
  &=s\frac{\sqrt{G_rG_uG_tA_rA_u}}{4\pi}\sum_{n=1}^{N}\frac{\sqrt{\tilde{F}_n}\theta_n}{d_{t,n}d_{r,n}}   e^{\frac{-j2\pi(d_{t,n}+d_{r,n})}{\lambda}}+\eta
\end{split}
\end{equation}
where $s$ is the transmitted signal by the BS and $\eta \sim \mathcal{CN}(0,\sigma^2)$ denotes the additive white Gaussian noise (AWGN) at the UE whose variance is $\sigma^2$.

Our target is to maximize the power of the UE received signal $P_r$ by designing the reflection coefficient vector $\bm{\theta}$ based on the estimated CSI. The problem can be formulated as
\begin{equation}\label{eq:ref-coef opt}
\check{\bm{\theta}} =\argmax\limits_{\bm{\theta}}\{|(\bm{g}\odot\bm{h})^T\bm{\theta}|^2\}.
\end{equation}
In principle, if the CSI is perfectly known, the optimal reflection coefficient of the $n$-th reflecting element of the RIS is
\begin{equation}\label{eq:opt solution of the ref-coef}
\check{\theta}_{n}=\frac{g_n^*h_n^*}{|g_nh_n|},n=1,\cdots,N.
\end{equation}

However, the CSI has to be acquired beforehand. Due to the huge number of RIS reflecting elements and the unavailability of signal processing units at the RIS, the CSI acquisition process is a challenging task. Moreover, the mobility of the UE make it imperative to acquire timely CSI. In the following section, we will show an efficient channel estimation method to address these issues.

\section{Proposed CSI acquisition Scheme based on 3D localization}\label{sec:proposed}
In this section, we show the localization-based channel estimation method.  The proposed approach capitalizes on two properties of mmWave channels:
1) The structural characteristics of the mmWave channel that is dominated by the LOS path; 
2) The different distances between the reflecting elements and the UE. 
Since the positions of the RIS and the BS are generally fixed, the BS-RIS channel is quasi-static and is assumed known. Our target is to estimate the RIS-UE channel 
by a two-step 3D positioning method, which is based on a novel concept named reflecting unit set (RUS), with an illusition shown in Fig. \ref{fig:RIS example}. These RUSs serve as \emph{anchor nodes} to locate the position of the UE. 
More details are given below. 

\subsection{Reflecting Unit Set}\label{sec:RUS}
\begin{figure}[!t]
\centering
\includegraphics[width=3.4in]{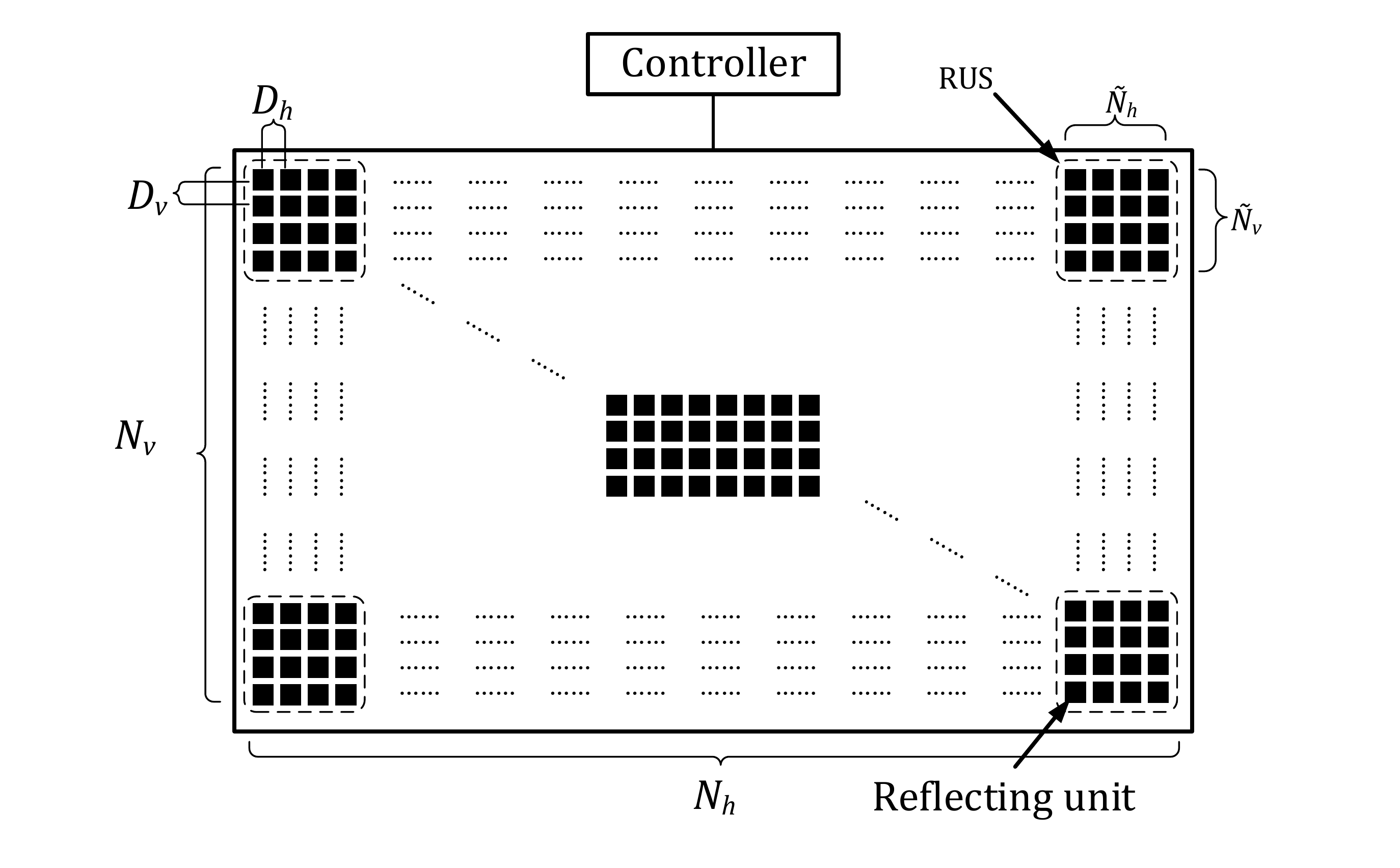}
\caption{The model of RIS with RUSs}
\label{fig:RIS example}
\end{figure}

RUS is a subset of the reflecting elements in a certain part of the RIS. An RUS can beamform the incident signal to the receiver, provided that a proper set of phase shifts is applied to the elements of the RUS. 
Without loss of generality, all $M$ RUSs are assumed to have the same shape, that is, a rectangular structure with $\tilde{N} = \tilde{N}_h\times \tilde{N}_v$ reflecting units, where $\tilde{N}_h$ and $\tilde{N}_v$ denote the numbers of columns and rows of the reflecting elements within one RUS, respectively.

The key properties of an RUS are explained follows:

\subsubsection{Active State}
An RUS is in the \emph{active} state when the reflecting elements inside the RUS reflect the incident signals. If the state is \emph{inactive}, on the other hand, the elements 
do not reflect the signal of interest, i.e., the electromagnetic waves are absorbed.  
Examples of metamaterials with these property are given in \cite{f_imani_perfect_2020,Li:2020programmable}. 
In addition, the elements of the RUS may be equipped with amplifiers in order to enhance the reflected signal\cite{Zhang:Active}. 

\subsubsection{Codebook and Codeword}
In order to beamform the incident wave to a certain direction, a codebook is introduced, which contains a set of codewords. Each codeword serves as a vector of the reflection coefficients.
Since the shape of the RUS is an UPA, we adopt the 3GPP codebook \cite{3gpp:38.214}, which is based on discrete Fourier transform (DFT) 
\begin{align}
\bm{u}_{p}&= \begin{bmatrix}1 \quad e^{j\frac{2\pi p}{O_1 \tilde{N}_v}} \quad \cdots \quad e^{j\frac{2\pi p(\tilde{N}_v-1)}{O_1\tilde{N}_v}} \end{bmatrix},\\
\bm{v}_{l,p}&=\begin{bmatrix}\bm{u}_p \quad e^{j\frac{2\pi l}{O_2\tilde{N}_h}}\bm{u}_p \quad \cdots \quad e^{j\frac{2\pi l(\tilde{N}_h-1)}{O_2\tilde{N}_h}}\bm{u}_p \end{bmatrix}^{T}.
\end{align}
The vectors $\bm{v}_{l,p}\in \mathbb{C}^{\tilde{N}_h\tilde{N}_v\times 1}$ are the codewords, where $p = 1, \cdots, \tilde{N}_v$ and $l = 1, \cdots, \tilde{N}_h$ are the indices of the vertical and horizontal beam directions, respectively. The positive integers $O_1$ and $O_2$ are the oversampling factors in the vertical and horizontal directions respectively.
We set the order of the reflecting units of the RUS as follows: starting from the lower left corner of the RUS, the number increases along the first column, then the second, etc, until the $\tilde{N}_h$-th column.
The relationship between the indices of the vertical and horizontal beam directions and the index of the codeword (denoted by $\tilde{n} \in [1,\tilde{N}]$) can be expressed as:
\begin{equation}\label{eq:index of codeword}
    \tilde{n} = (l-1) \tilde{N}_v+p,
\end{equation}
and then the $\tilde{n}$-th codeword can be defined as
$\tilde{\bm{\theta}}_{\tilde{n}} \triangleq\bm{v}_{l,p}$.
The corresponding codebook, denoted by $\tilde{\bm{\Theta}}$, can be expressed as
\begin{equation}\label{eq:codebook-RUS}
  \tilde{\bm{\Theta}} = [\tilde{\bm{\theta}}_1,...,\tilde{\bm{\theta}}_{\tilde{N}}].
\end{equation}

\subsubsection{Selection of Codeword}
The RUS can be seen as a small RIS. According to (\ref{eq:BS-RIS-UE channel}), the cascaded BS-RUS-UE channel  corresponding to the $\tilde{n}$-th codeword is expressed as
\begin{equation}\label{eq:y-mth RUS model}
  \tilde{w}_{\tilde{n}} =  \sqrt{G_rG_uG_t}(\tilde{\bm{g}}\odot\tilde{\bm{h}})^T\tilde{\bm{\theta}}_{\tilde{n}},
\end{equation}
where $\tilde{\bm{h}}$ and $\tilde{\bm{g}}$, derived from (\ref{eq:BS-n channel})-(\ref{eq:BS-RIS channel}) and (\ref{eq:UE-n channel})-(\ref{eq:UE-RIS channel}) respectively, denote the BS-RUS channel and the RUS-UE channel.
Based on (\ref{eq:y-signal recevied}), the received signal at the UE via the RUS is
\begin{equation}\label{eq:received signal via RUS}
  \tilde{r}_{\tilde{n}}=s\tilde{w}_{\tilde{n}}+\eta.
\end{equation}
The codeword that maximizes the power of the UE received signal is chosen, and the index of the selected codeword can be expressed as
\begin{equation}\label{eq:codeword selection}
  \tilde{n}_{opt} = \argmax\limits_{\tilde{n}}\{\tilde{r}_{\tilde{n}}\}.
\end{equation}
Note that several trials are needed to find the optimal codeword. An ordinary strategy is traversing all codewords in the codebook $\tilde{\bm{\Theta}}$ to select the codeword. The same selected codeword can be used in a certain time interval, and even re-used by all the RUSs. This will help reduce the complexity at the cost of a weaker pilot signal reception at the UE. 

Note also that the shape of the RUS is not limited to UPA, and the size of the RUS can be customized. Smaller RUS may provide better accuracy when the reflected signal is strong enough for the UE to acquire good channel estimation quality. However, when the desired reflected signal is weak, for example due to large pathloss, the area of the RUS should be larger, or the RUSs should be equipped with power amplifiers. 

We then estimate the RUS-UE distances based on the angle-delay structure of the wideband mmWave channels. 

\subsection{RUS-UE Distance Estimations}\label{sec:multi-subbands}
Due to the sparse distribution of the multipath delays, the channel covariance matrix in the frequency domain has a low rank. In this paper, we exploit this low-rankness property to increase the accuracy of the delay estimation. 
With ultra-wide bandwidth, mmWave communications can offer higher resolution of the delay and improve the accuracy of the distance estimation. In this section,  we propose a distance estimation method exploiting the low rank of the frequency-domain channel covariance matrix.

Denote the central frequency by $F_c$. The whole bandwidth contains $K$ sub-bands, where the center frequency and the wavelength of the $k$-th sub-band are $f_k$ and $\lambda_k=\frac{c}{f_k}$ respectively ($c$ denotes the speed of light).
The vector of frequencies is denoted as $\bm{f}=[f_1,...,f_K]$.
The power spectral density of the noise is $\bar{\sigma}^2$, and the power of the noise at each sub-band is $\sigma_{n}^2 = \bar{\sigma}^2 f_d$, where $f_d$ is the bandwdith of each sub-band. 
Without loss of generality, we assume equal power allocation in all sub-bands, and the transmitted power of a sub-band is $\sigma_t^2 = P_t/K$, where $P_t$ is the total transmit power. 

The received signal that is reflected by the $m$-th RUS (denoted as $\text{RUS}_m$) at sub-band $k$ is 
\begin{equation}\label{eq:y-mk}
  \tilde{r}_{m,k} = s_{m,k}\tilde{w}_{m,k}+\eta_{m,k},
\end{equation}
where $\tilde{s}_{m,k}$, $\tilde{w}_{m,k}$ and $\eta_{m,k}\sim \mathcal{CN}(0,\sigma_{n}^2)$ denote the transmitted signal, the cascaded channel via $\text{RUS}_m$ at $k$-th sub-band and the AWGN at the UE, respectively.
We show more details of several distance estimation methods  below.

\subsubsection{Matrix Pencil (MP)-based Method}

According to the low-rank property and the wideband setting, we develop a Matrix Pencil \cite{Hua:1990matrixpencil} based distance estimation  method. Considering the cascaded channel via the $m$-th RUS with $\tilde{N}$ units and the proper codeword, the relationship between the wavelength of the $K$ sub-bands can be expressed as
\begin{equation}\label{eq:MP-wavelength}
  \lambda_k = \frac{c}{f_1+(k-1)f_d},k=1,\ldots,K.
\end{equation}
Then based on the one-dimensional MP method\cite{Abdo:2015TwoDMP}, the estimated distance between the UE and the $m$-th RUS can be expressed as

\begin{equation}\label{eq:MP-d}
    \hat{d}_m = \frac{1}{\tilde{N}}\sum_{\tilde{n}=1}^{\tilde{N}}\frac{c\log \ddot{z}_{\tilde{n}}}{-j2\pi f_d}-d_{bs,m}
\end{equation}
where
\begin{equation}\label{eq:MP-zn}
  \ddot{z}_{\tilde{n}} = e^{\frac{-j2\pi(d_{t,\tilde{n}}+d_{r,\tilde{n}})f_d}{c}}, \tilde{n}=1,\dots,\tilde{N},
\end{equation}
denotes the eigenvalue of the Hankel matrix product for matrix pencil method 
and $d_{bs,m}$ denotes the BS-$\text{RUS}_m$ distance. 

\subsubsection{Conventional MMSE-based Method}\label{sec:MMSE-independ}
Adopting the conventional MMSE-based estimation criterion, we show the MMSE channel estimation independently for each sub-band, when the $\text{RUS}_m$ is activated and a proper codeword is chosen.

Based on (\ref{eq:y-mk}), with the MMSE method, the estimated cascaded channel of the $k$-th sub-band via the $m$-th RUS is expressed as
\begin{equation}\label{eq:H-mk MMSE}
  \hat{w}_{m,k} = R_{m,k}s_{m,k}^*(\sigma_t^2R_{m,k}+\sigma_n^2)^{-1}\tilde{r}_{m,k},
\end{equation}
where $R_{m,k} = \mathbb{E}[\tilde{w}_{m,k}\tilde{w}_{m,k}^*]$ denotes the covariance coefficient of the cascaded channel of the $k$-th sub-band and the $m$-th RUS.
The mean squared error (MSE) of the estimation is
\begin{equation}\label{eq:Hm-MSE }
\begin{split}
\mathcal{M}_{m,k}
& = \mathbb{E}\{\|\hat{w}_{m,k}-\tilde{w}_{m,k}\|^2\} \\
& = R_{m,k}(1+\frac{\sigma_t^2}{\sigma_n^2}R_{m,k})^{-1}\\
& = \frac{R_{m,k}\sigma_n^2}{\sigma_t^2R_{m,k}+\sigma_n^2}
\end{split},
\end{equation}

Based on (\ref{eq:H-mk MMSE}), the estimated cascaded channel vector of the $m$-th RUS is expressed as
\begin{equation}\label{Eq:Hm est inden}
\hat{\bm{w}}_m = [\hat{w}_{m,1}, \hat{w}_{m,2}, \ldots, \hat{w}_{m,K}]^T.
\end{equation}
and the distance between the UE and the $\text{RUS}_m$ can be acquired by Bartlett's method\cite{EH:2012bartlett}.

\subsubsection{Proposed Joint MMSE (JMMSE)-based Method}
In order to give full play to the role of ultra-wideband, we develop a JMMSE-based distance estimation method which is able to exploit the low-rankness structure of the covariance matrix in the frequency domain\cite{yin:13a}. 
When $\text{RUS}_m$ is in the active state with a proper codeword selected, we estimate the corresponding RUS-UE distance with the JMMSE channel estimation of all sub-bands.

The BS-$\text{RUS}_m$-UE channel  at all $K$ sub-bands is written as
\begin{equation}\label{Eq:channel broadband}
\tilde{\bm{w}}_m = [\tilde{w}_{m,1}, \tilde{w}_{m,2}, \ldots, \tilde{w}_{m,K}]^T,
\end{equation}
where $\tilde{w}_{m,k}(k = 1, \cdots, K)$ denotes the BS-$\text{RUS}_m$-UE channel at sub-band $k$.

The received signal via $\text{RUS}_m$ at all $K$ sub-bands is
\begin{equation}\label{eq:Ym - received signal with m-th RUS}
  \tilde{\bm{r}}_m = [\tilde{r}_{m,1},\ldots,\tilde{r}_{m,K}]^T = \bm{S}_m\tilde{\bm{w}}_m + \bm{\eta}_m,
\end{equation}
where  $\bm{S}_m = \text{diag}([s_{m,1},\ldots,s_{m,K}]^T)$ and $\bm{\eta}_m=[\eta_{m,1},\ldots,\eta_{m,K}]^T$ denote the transmitted signal vector of all $K$ sub-bands via $\text{RUS}_m$ and the noise vector respectively. 

With the JMMSE estimation method, the estimated cascaded channel via $\text{RUS}_m$ is given by
\begin{equation}\label{eq:Hm-Channel MMSE est}
   \hat{\bm{w}}_m  =\bm{R}_m(\bm{R}_m+\frac{\alpha\sigma_n^2}{\sigma_t^2}\bm{I})^{-1}\bm{S}_m^{-1}\tilde{\bm{r}}_m
\end{equation}
where $\bm{R}_m=\mathbb{E}[\tilde{\bm{w}}_m\tilde{\bm{w}}_m^H]$ is the frequency-domain channel covariance matrix and $\alpha$ is a parameter of regularization.
Denote the true path delay by $t_m$, the estimated path delay of the BS-$\text{RUS}_m$-UE channel is computed by
\begin{equation}\label{eq:delay of BS-RUS-UE}
\begin{split}
   \hat{t}_m &=\argmax_{t_m}\{\|\bm{\hat{w}}_m^T\bm{b}(t_m)\|^2\}\\
     &= t_m + \Delta t_m ,
\end{split}
\end{equation}
where $\Delta t_m$ is the estimation error and $\bm{b}(t)\in\mathbb{C}^{K\times 1}$ is the delay response vector defined as
\begin{equation}\label{eq:b(t)}
\bm{b}(t)=[e^{j2\pi f_1t},\cdots,e^{j2\pi f_Kt}]^T.
\end{equation}

The distance between the UE and $\text{RUS}_m$ is $d_m=c (t_m-t_{bs,m})$, and the estimated one is expressed as
\begin{equation}\label{eq:d-UE-RUS}
\begin{split}
   \hat{d}_{m} &= c(\hat{t}_m - t_{bs,m}) \\
     &= d_m + \Delta d_m,
\end{split}
\end{equation}
where $\Delta d_m=c \Delta t_m$ denotes the estimated error of the distance of UE-$\text{RUS}_m$ channel,
and $t_{bs,m}$ denotes the true delay of BS-$\text{RUS}_m$ channel. Since the BS and the RIS are generally fixed, their locations and the delay $t_{bs,m}$ can be assumed known.

\subsection{Proposed Coplanar ML-based 3D Localization Method}\label{sec:localization-alg}
Since all the anchors (i.e., RUSs) of the RIS are on the same plane (y-z plane), most trilateration localization methods fail to work in this case, such as the least squares matrix based method \cite{KW:2004LSforToA}.
The authors of \cite{Chan:1994simple} and \cite{Chan:2006exact} gave alternate ML-based solutions for TDOA and TOA estimation in 2D settings with col-linear anchors.
In this section, we propose a novel coplanar maximum likelihood-based (CML) 3D localization method which is dedicated to channel estimation in RIS-aided mmWave communication systems.

Let $\bm{p} = (x,y,z)$ denote the position of the UE. The vector of the estimated distances between all $M$ RUSs and the UE is written as
\begin{equation}\label{eq:d-est vector}
   \hat{\bm{d}}  = [\hat{d}_1,\hat{d}_2,\ldots,\hat{d}_M]^T = \bm{d} + \bm{e}, \\
\end{equation}
where $\bm{d} = [d_1,d_2,\ldots,d_M]^T$ is the vector of the corresponding true distances, and $\bm{e} = [e_1,\ldots,e_M]^T $ is the vector of the additive estimation errors. 

Based on (\ref{eq:Ym - received signal with m-th RUS}) and (\ref{eq:Hm-Channel MMSE est}), the estimated error of the BS-$\text{RUS}_m$-UE channel is
\begin{equation}\label{eq:err_mChannel}
\begin{split}
   \bm{e}_{wm} & = \hat{\bm{w}}_m-\bm{w}_m \\
     & =(\bm{A}- \bm{I})\bm{w}_m + \bm{A}\bm{S}^{-1}\bm{\eta}_m
\end{split}
\end{equation}
where $\bm{A} = \bm{R}_m(\bm{R}_m+\frac{\alpha\sigma_n^2}{\sigma_t^2}\bm{I})^{-1}$.
When the SNR is high, $\bm{A}$ is close to $\bm{I}$, and then $\bm{e}_{wm}$ follows a zero-mean Gaussian distribution. Based on this fact, the entry of $\bm{e}$ is modeled as a Gaussian distribution $e_m\sim \mathcal{CN}(0,\sigma_{d}^2),m=1,\ldots, M$ where $\sigma_d^2$ is the variance. The conditional probability density function (PDF) of the estimated distance $\hat{\bm{d}}$ is thus
\begin{equation}\label{eq:pdf of d est}
  P(\hat{\bm{d}}|\bm{p}) = (2\pi \sigma_d^2)^{-\frac{M}{2}}\text{exp}\bigg(-\frac{J}{2\sigma_d^2}\bigg),
\end{equation}
where
\begin{equation}\label{eq:J-ML}
  J=\sum_{m=1}^{M}(d_m-\hat{d}_m)^2.
\end{equation}
According to the maximum likelihood criterion, finding $\bm{p}$ from $\hat{\bm{d}}$ is to equivalent to minimizing (\ref{eq:J-ML}).

The coordinates of the four given coplanar RUSs are $\bm{p}_1=(0,0,0)$, $\bm{p}_2=(0,a,0)$, $\bm{p}_3=(0,a,b)$ and $\bm{p}_4=(0,0,b)$, where $a$ and $b$ are real. The estimated distances between each anchor node and the UE to be localized are $\hat{d}_m(m=1,\ldots,4)$. Then, the true distances are calculated by

\begin{align}
  d_1 & =(x^2 + y^2 + z^2)^{\frac{1}{2}}, \label{eq:d-1}\\
  d_2 & =[x^2 + (y-a)^2 + z^2]^{\frac{1}{2}}, \label{eq:d-2}\\
  d_3 & =[x^2 + (y-a)^2 + (z-b)^2]^{\frac{1}{2}}, \label{eq:d-3}\\
  d_4 & =[x^2 + y^2 + (z-b)^2]^{\frac{1}{2}}.\label{eq:d-4}
\end{align}
The closed-form solutions for minimizing (\ref{eq:J-ML}) are shown in Theorem \ref{theo:ML}.

\begin{theorem}\label{theo:ML}
The maximum likelihood (ML) estimation of the distances are
\begin{align}
  d_1 & =\frac{\hat{d}_1\bigg (\hat{d}_1^2+\hat{d}_3^2\pm\sqrt{(\hat{d}_1^2+\hat{d}_3^2)(\hat{d}_2^2+\hat{d}_4^2)} \bigg)}{2(\hat{d}_1^2+\hat{d}_3^2)}, \label{eq:d-ml-1}\\
  d_2 & =\frac{\hat{d}_2d_1}{2d_1-\hat{d}_1},\label{eq:d-ml-2} \\
  d_3 & = \frac{\hat{d}_3d_1}{\hat{d}_1},\label{eq:d-ml-3}\\
  d_4 & = \frac{\hat{d}_4d_1}{2d_1-\hat{d}_1},\label{eq:d-ml-4}
\end{align}
\end{theorem}
\begin{proof}
To minimize (\ref{eq:J-ML}), we differentiate $J$ with respect to $x$, $y$ and $z$ and set the results to zero
\begin{align}
  \frac{\partial J}{\partial x} & = \sum_{m=1}^{4}\frac{2x(d_m-\hat{d}_m)}{d_m} \label{eq:J_diff x}\\
  \frac{\partial J}{\partial y} & = \sum_{m=1}^{4}\frac{2y(d_m-\hat{d}_m)}{d_m}-\sum_{m=2,3}\frac{2a(d_m-\hat{d}_m)}{d_m} \label{eq:J_diff y} \\
  \frac{\partial J}{\partial z} & = \sum_{m=1}^{4}\frac{2z(d_m-\hat{d}_m)}{d_m}-\sum_{m=2,4}\frac{2b(d_m-\hat{d}_m)}{d_m} \label{eq:J_diff z}
\end{align}
Combining (\ref{eq:J_diff x}) - (\ref{eq:J_diff z}), $d_2,d_3\ \text{and}\ d_4$ can be formulated in terms of $d_1$
\begin{equation}\label{eq:d_234}
  d_2  = \frac{\hat{d}_2 d_1}{2d_1-\hat{d}_1},\ d_3  = \frac{\hat{d}_3 d_1}{\hat{d}_1},\ d_4 = \frac{\hat{d}_4 d_1}{2d_1-\hat{d}_1}
\end{equation}

From (\ref{eq:d-1})-(\ref{eq:d-4}), the relationship among $d_1, d_2,d_3\ \text{and}\ d_4$ can be expressed as
\begin{equation}\label{eq:d total}
  d_1^2+d_3^2-d_2^2-d_4^2=0.
\end{equation}
Then, substituting (\ref{eq:d_234}) into (\ref{eq:d total}), we have
\begin{equation}\label{eq:d1}
  \frac{4K_{13}d_1^4}{\hat{d}_1^2}-\frac{4K_{13}d_1}{\hat{d}_1^3}+(K_{13}-K_{24})d_1^2=0
\end{equation}
where $K_{13}=\hat{d}_1^2+\hat{d}_3^2>0$ and $K_{24}=\hat{d}_2^2+\hat{d}_4^2>0$.
Since $d_m>0$, (\ref{eq:d1}) can be solved by
\begin{equation}\label{eq:d1_result}
\begin{split}
   d_1 & =\frac{\hat{d}_1\big (K_{13}\pm\sqrt{K_{13}K_{24}} \big)}{2K_{13}} \\
     & =\frac{\hat{d}_1\bigg (\hat{d}_1^2+\hat{d}_3^2\pm\sqrt{(\hat{d}_1^2+\hat{d}_3^2)(\hat{d}_2^2+\hat{d}_4^2)} \bigg)}{2(\hat{d}_1^2+\hat{d}_3^2)}.
\end{split}
\end{equation}

Finally, $d_2,d_3,d_4$ can be derived from (\ref{eq:d_234}) and (\ref{eq:d1_result}).
\end{proof}

Since at least one root of $d_1$ is positive, we choose the correct root as follows. If only one root is positive, it is the value we needed. If all roots are positive, the one corresponding to the smaller $J$ in (\ref{eq:J-ML}) is selected.

According to (\ref{eq:d-1})-(\ref{eq:d-4}), Theorem \ref{theo:ML}, and Euclidean geometry, the  3D position $\bm{p}=  (x, y, z)$ of the UE can be solved by
\begin{align}
     y & =  \frac{\bar{d}_{1}^2-\bar{d}_{2}^2-\bar{d}_{3}^2+\bar{d}_{4}^2+2a^2}{4a},\label{eq:p-esty}\\
     z & = \frac{\bar{d}_{1}^2+\bar{d}_{2}^2-\bar{d}_{3}^2-\bar{d}_{4}^2+2b^2}{4b} ,\label{eq:p-estz} \\
     x & = \sqrt{\bar{d}_{m}^2-\big(y-\bm{p}_{m}(y)\big)^2-\big(z-\bm{p}_{m}(z)\big)^2},\label{eq:p-estx}
\end{align}
where $m=1,2,3,4$. We assume the UE and the BS are in the front side of the RIS, which means $\bar{x}>0$, as shown in Fig. \ref{fig:RIS-assisted_Communication_System}. The CML method selects the solution of $\bm{p} = (x,y,z)$ that gives the minimum $J$.

Note that the constraints on the positions of the RUSs can be relaxed, i.e., they are not necessarily on the four corners of the RIS, and our method can be extended easily if we transform the basis of the y-z plane according to the properties of the linear transformation space. 

The closed-form Fisher Information Matrix and the Cramer-Rao lower bounds (CRLBs) of our 3D positioning method are shown in Theorem \ref{theo:CRLB}.

\begin{theorem}\label{theo:CRLB}
The elements of Fisher Information Matrix (FIM) denoted by $\bm{\Psi}$  is expressed as
\begin{equation}\label{eq:FIM-elements-2}
  \begin{split}
     \bm{\Psi}_{ij} & = -\mathbb{E}\bigg[\frac{\partial^2\ln P(\hat{\bm{d}}|\bm{p})}{\partial\bm{p}(i)\partial\bm{p}(j)}\bigg] \\
       & =\sum_{m=1}^{M}\frac{1}{\sigma_d^2}\frac{\big(\bm{p}(i)-\bm{p}_{m}(i)\big)\big(\bm{p}(j)-\bm{p}_{m}(j)\big)}{d_m^2},
  \end{split}
\end{equation}
where $i,j=1,2,3,$ and $\bm{p}_m=(x_m,y_m,z_m)$ denotes the location of the $m$-th anchor node.

Then, the CRLB matrix of the CML 3D localization method is
\begin{equation}\label{eq:CRLB}
  \bm{C}=
  \left [\begin{matrix}
    C_{11} & C_{12} & C_{13} \\
    C_{21} & C_{22} & C_{23} \\
    C_{31} & C_{32} & C_{33}
  \end{matrix}\right ]=\bm{\Psi}^{-1},
\end{equation}
where $C_{11}$ , $C_{22}$ and $C_{33}$ denote the CRLBs of the estimation of the (x, y, z) coordinates respectively.
\end{theorem}
\begin{proof}
See Appendix \ref{appendix:CRLB-proof}.
\end{proof}

Overall,  with the  position of the UE estimated by our localization scheme that includes the distance estimation phase and CML-based 3D localization phase, we may obtain the estimation of the channel $\bm{g}$ and further compute the reflecting coefficient vector $\bm{\theta}$, according to Sec. \ref{sec:model}.
Taking the JMMSE-based method as an example, we summarize the procedure in Algorithm \ref{alg:CE}. 

\begin{algorithm}[!t]
\renewcommand{\algorithmicrequire}{\textbf{Input:}} 
\renewcommand{\algorithmicensure}{\textbf{Output:}} 
\caption{The Proposed JMMSE and CML based Channel Estimation Scheme}
\label{alg:CE}
\begin{algorithmic}[1]
\Require $\bm{h},\bm{S},\bm{p}_{bs},\bm{p}_{m},m=1,\ldots,M,\bm{p}_n,n=1...,N$
\For{$m=1,\ldots,M$}
\State Activate the $m$-th RUS to reflect the transmitted signal;
\State \% Choose the optimal codeword of the $m$-th RUS.
\For{$\tilde{n}=1,\ldots,\tilde{N}$}
\State Select the $\tilde{n}$-th codeword $\tilde{\bm{\theta}}_{\tilde{n}}$ for the RUS;
\State UE receives the wideband pilot signal:
\State  $\tilde{\bm{r}}_{\tilde{n}}=[\tilde{r}_{\tilde{n},1},\ldots,\tilde{r}_{\tilde{n},K}]^T$
\EndFor
\State $\tilde{n}_{opt} = \argmax\limits_{\tilde{n}}(\|\tilde{\bm{r}}_{\tilde{n}}\|) $;
\State Estimate the BS-$\text{RUS}_m$-UE wideband channel based on JMMSE (\ref{eq:Hm-Channel MMSE est});
\State Estimate the channel delay based on (\ref{eq:delay of BS-RUS-UE}) and (\ref{eq:b(t)});
\EndFor
\State \% Localization of the UE with the CML method.
\State Determine the position of the UE $\bm{p}$ based on Theorem \ref{theo:ML} and (\ref{eq:p-esty})-(\ref{eq:p-estx});
\State \% Channel reconstruction.
\State Calculate the distance between the UE and the elements of the RIS: $d_{r,n} = \|p-p_{n}\|,n=1,...,N$;
\State Reconstruct the RIS-UE channel $\bm{g}$ based on (\ref{eq:UE-RIS channel});
\Ensure $\bm{p}$, $\bm{g}$
\end{algorithmic}
\end{algorithm}

Note that the radiation patterns are not needed when computing the optimal reflection coefficients $\bm{\theta}$, only the path delay $d_{r,n} = \|p-p_{n}\|,n=1,...,N$ is sufficient. 
The complexity order of our method is $\mathcal{O}(M\tilde{N})$, which does not depend on the total number of the reflecting units on the RIS. It helps to keep the complexity low, especially for a large RIS. 
As the number of the reflecting elements increases, the complexity and training overhead of most traditional methods will increase significantly.
For a large RIS, our proposed scheme has a better performance due to the exploitation of the differences of the RUS-UE distances. 
In reality, the number of the RUSs can be set to $M=4$ according to our CML method,  and the number of reflecting units on each RUS is much smaller than that on the overall RIS, i.e., ($\tilde{N}\ll N$). 

Consider for example an RIS with four RUSs, each of the RUS has four rows and four columns of reflecting units. A DFT codebook containing 16 codewords is used. Assuming the same codeword is shared by all RUSs, the searching of the codeword takes 16 OFDM symbols. Four more symbols are needed to complete activating all RUSs. Thus the total number of OFDM symbols needed can be limited to 20. In general, the above proposed method can be completed efficiently with low time-frequency resource utilization.

\section{Numerical Results}\label{sec:simulation}
In this section, we verify the proposed method in simulations and analyze its performance and robustness.
We consider an RIS-aided mmWave communication system as shown in Fig. \ref{fig:RIS-assisted_Communication_System}.
The coordinates of the BS are $(5,-5,2)$ (meters), and the LOS path between the BS and the UE is assumed to be negligible due to the presence of blocking objects. The main simulation parameters are shown in Table \ref{tab:paras}. 
\begin{table}
  \centering
   \caption{Simulation Parameters}\label{tab:paras}
  \begin{tabular}{|c|c|l|}
    \hline
    \textbf{Symbol} & \textbf{Value} & \textbf{Definition} \\
    \hline
    $N$ & 8192 & Number of reflecting units on the RIS \\
    $N_v$ & 64 & Number of rows of reflecting units on the RIS \\
    $N_h$ & 128 & Number of columns of reflecting units on the RIS \\
    $D_v$(m) & 0.005 & Row spacing between adjacent units \\
    $D_h$(m) & 0.005 & Column spacing between adjacent units \\
    $M$ & 4 & Number of RUSs on the RIS \\
    $\tilde{N}$ & 16 & Number of reflecting units on an RUS\\
    $\tilde{N}_v$ & 4 & Number of rows of reflecting units on an RUS\\
    $\tilde{N}_h$ & 4 & Number of columns of reflecting units on an RUS\\
    $F_c$(GHz) & 28 & Center frequency\\
    $f_d$(Mhz) & 3.6 &Bandwidth of a subband (5 resource blocks)\\
    $K$ & 128 & Number of subbands\\
    $\bar{\sigma}^2$(dBm/Hz) & -170 & Noise power spectral density\\
    $\alpha$ & $10^4$& Coefficient of regularization\\
    \hline
  \end{tabular}
\end{table}
The RIS is located in the y-z plane with the center at $\bm{p}_c=(0,0.32,0.16)$. The RIS contains $N=N_v\times N_h$ reflecting units.
$M=4$ RUSs are at the four corners of RIS with size $\tilde{N}=\tilde{N}_h\times \tilde{N}_v$, as shown in Fig. \ref{fig:RIS example}.
The total transmit power is $P_t=30$ dBm. The gain of each reflecting unit is $G_u=9.03$ dBi\cite{Tang:2020wireless}. The antenna gains of the transmit antenna and  receive antenna are $G_t=G_r=21$ dBi, the aperture of the receive antenna $A_r =\frac{\lambda^2}{4\pi}$,
and the radiation patterns of the transmit and receive antennas are $F_n^{tx}=F_n^{rx}=1, n=1,\ldots,N$.
The radiation pattern of each reflecting unit is:
\begin{equation}\label{eq:radiation-pattern}
  F(\vartheta,\varphi)=  \left\{
  \begin{aligned}
    \cos^3(\vartheta) & \quad \vartheta \in  [0,\frac{\pi}{2}],\varphi \in [0,2\pi] \\
                    0 & \quad  \vartheta\in  (0,\frac{\pi}{2}],\varphi \in [0,2\pi]
  \end{aligned}\right.
\end{equation}

\begin{figure}[!t]
\centering
\includegraphics[width=3.2in]{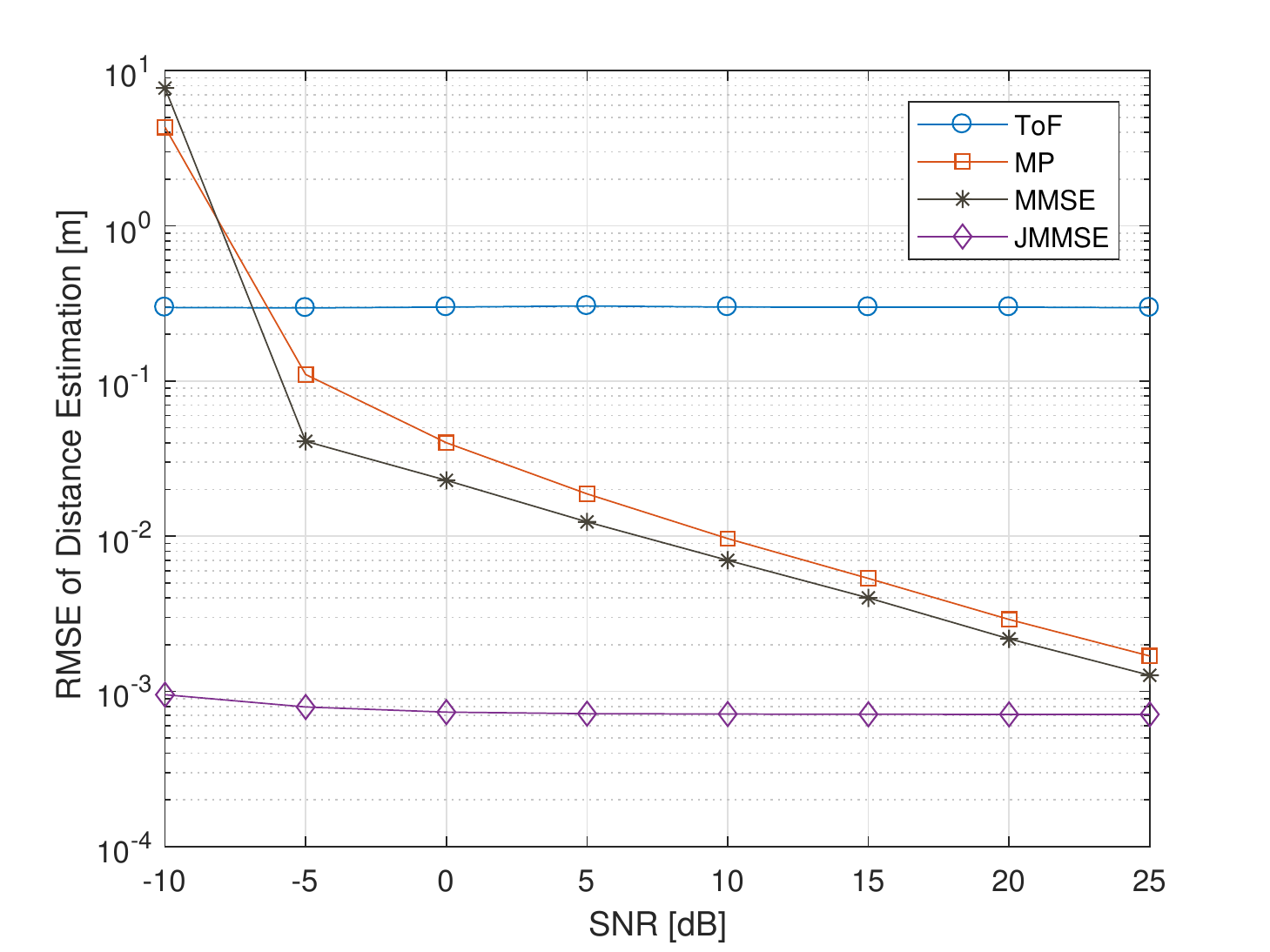}
\caption{The RMSE performance of distance estimation vs SNR of the received signal via the RUSs.}
\label{fig:d_snr}
\end{figure}
For a fixed UE with coordinate $\bm{p}=(5,0.32,0.16)$, 
Fig. \ref{fig:d_snr} shows the root mean squared error (RMSE) of the estimated distances, defined as $\sqrt{\mathbb{E}[\|\bm{d}-\hat{\bm{d}}\|^2]}$, versus the signal to noise radio (SNR) of the received signal via the RUSs, which is defined as
\begin{equation}\label{eq:SNR}
  \text{SNR} = \frac{\|\bm{S}\tilde{\bm{w}}\|^2}{\|\bm{\eta}\|^2} = \frac{P_t\|\tilde{\bm{w}}\|^2}{\bar{\sigma}^2 Kf_d},
\end{equation}
according to (\ref{eq:Ym - received signal with m-th RUS}). It varies as a function of the total transmit power $P_t$.
The distance estimated by the conventional method based on Time of Flight (ToF) is drawn from the Gaussian distribution $\mathcal{CN}(0,\sigma_{\text{ToF}}^2)$ where $\sigma_{\text{ToF}}^2$ is  1 ns\cite{Lemic:2016localization}.
For the estimated distance, it is observed that the RMSE under the proposed JMMSE method is much smaller than the other three schemes. The RMSEs under MMSE and MP are both smaller than the ToF scheme when the SNR is greater than -5 dB. 

\begin{figure}[!t]
\centering
\includegraphics[width=3.2in]{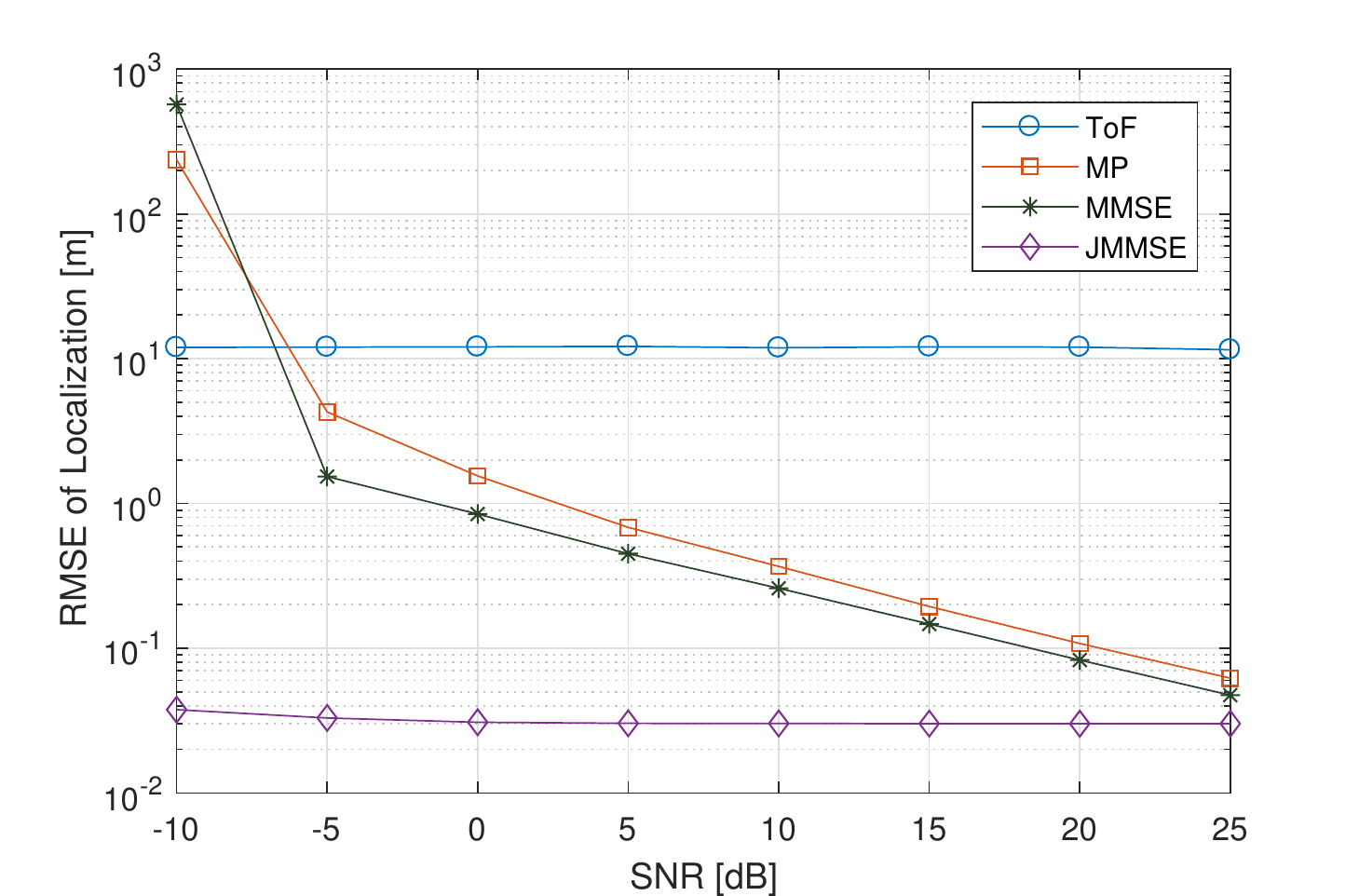}
\caption{The RMSE performance of the localization vs SNR of the received signal via the RUSs.}
\label{fig:p_snr}
\end{figure}
Combining the aforementioned distance estimation methods and the CML localization method, Fig. \ref{fig:p_snr} shows the RMSE of the localization of the UE, which is defined as $\sqrt{\mathbb{E}[\|\bm{p}-\hat{\bm{p }}\|^2]}$. Similar to the result of the distance estimation, the RMSE obtained by combining JMMSE and CML is much smaller than the other three schemes and achieves centimeter level positioning accuracy.

\begin{figure}[!t]
\centering
\includegraphics[width=3.2in]{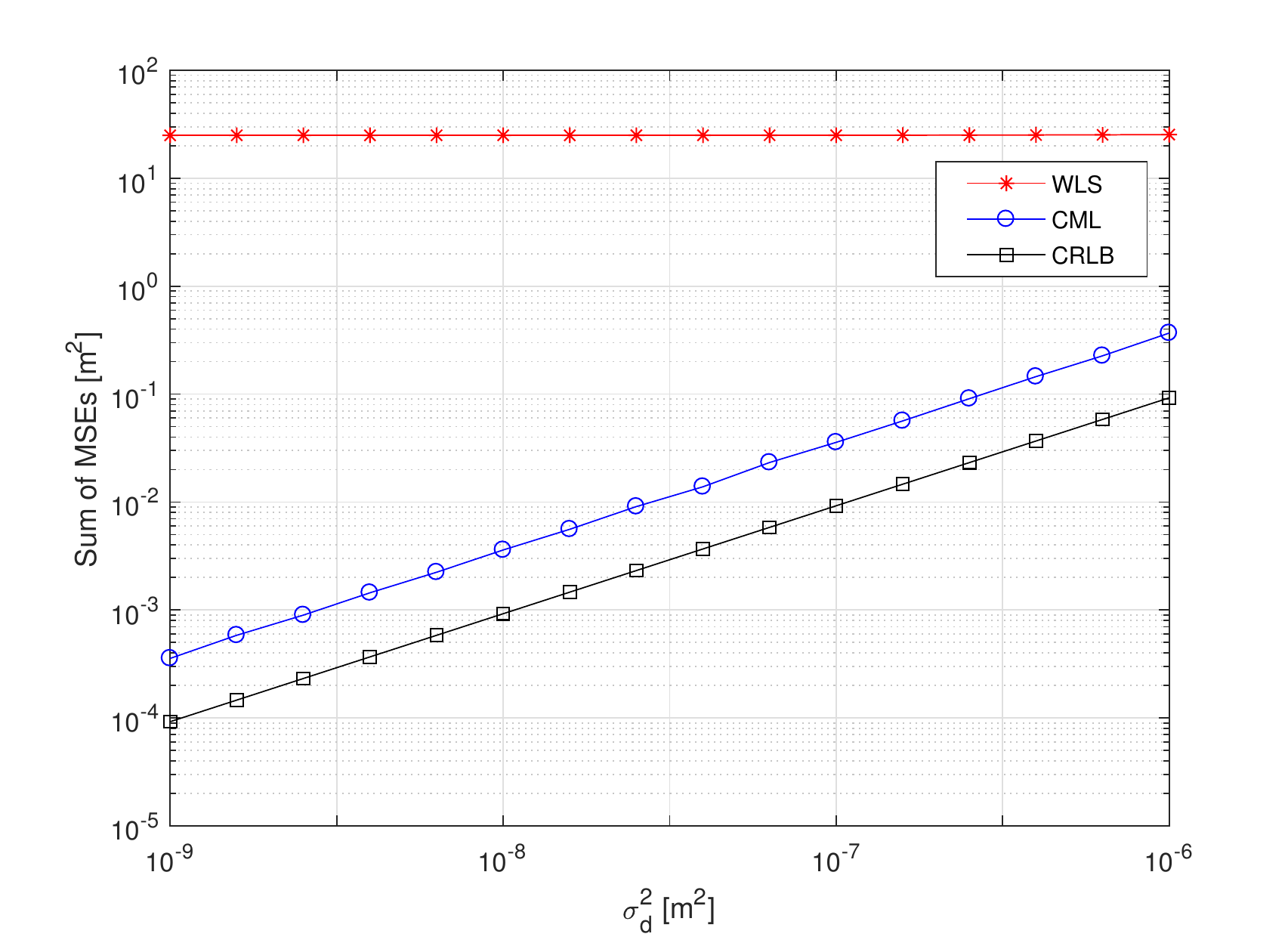}
\caption{The MSE performance of the 3D coplanar localization method versus the distance error $\sigma_d^2$.}
\label{fig:sm}
\end{figure}
To analyze the performance of the proposed CML localization algorithm theoretically, we consider a UE at a fixed position $\bm{p}=(5,0.32,0.16)$, which is unknown to the system, and assume that all the estimated errors of the $\text{UE-RUS}_m$ distances share the same noise variance $\sigma_d^2$. 
That is, each error of the measured distance is drawn from $\mathcal{CN}(0,\sigma_{d}^2)$. 
The sum of the MSEs is defined as
\begin{equation}\label{eq:sum of MSE}
  \tilde{\mathcal{M}} = \mathcal{M}(x)+\mathcal{M}(y)+\mathcal{M}(z),
\end{equation}
where $\mathcal{M}(x)$, $\mathcal{M}(y)$ and $\mathcal{M}(z)$ are the MSEs of the estimated x, y, z-coordinate of the UE respectively, e.g.,
$\mathcal{M}(x)=\mathbb{E}[(\hat{x}-\bar{x})^2]$. According to Theorem \ref{theo:CRLB}, the sum of CRLBs is defined as
\begin{equation}\label{eq:total CRLB}
  \tilde{\mathcal{C}} = C_{11}+C_{22}+C_{33}.
\end{equation}
Fig. \ref{fig:sm} compares the sum of the MSE for the CML method with the corresponding CRLB and the traditional weighted least squares (WLS) algorithm \cite{ni_uwb_2019}. The figure shows the proposed CML is much closer to the CRLB than the traditional WLS.

\begin{figure}[!t]
\centering
\includegraphics[width=3.2in]{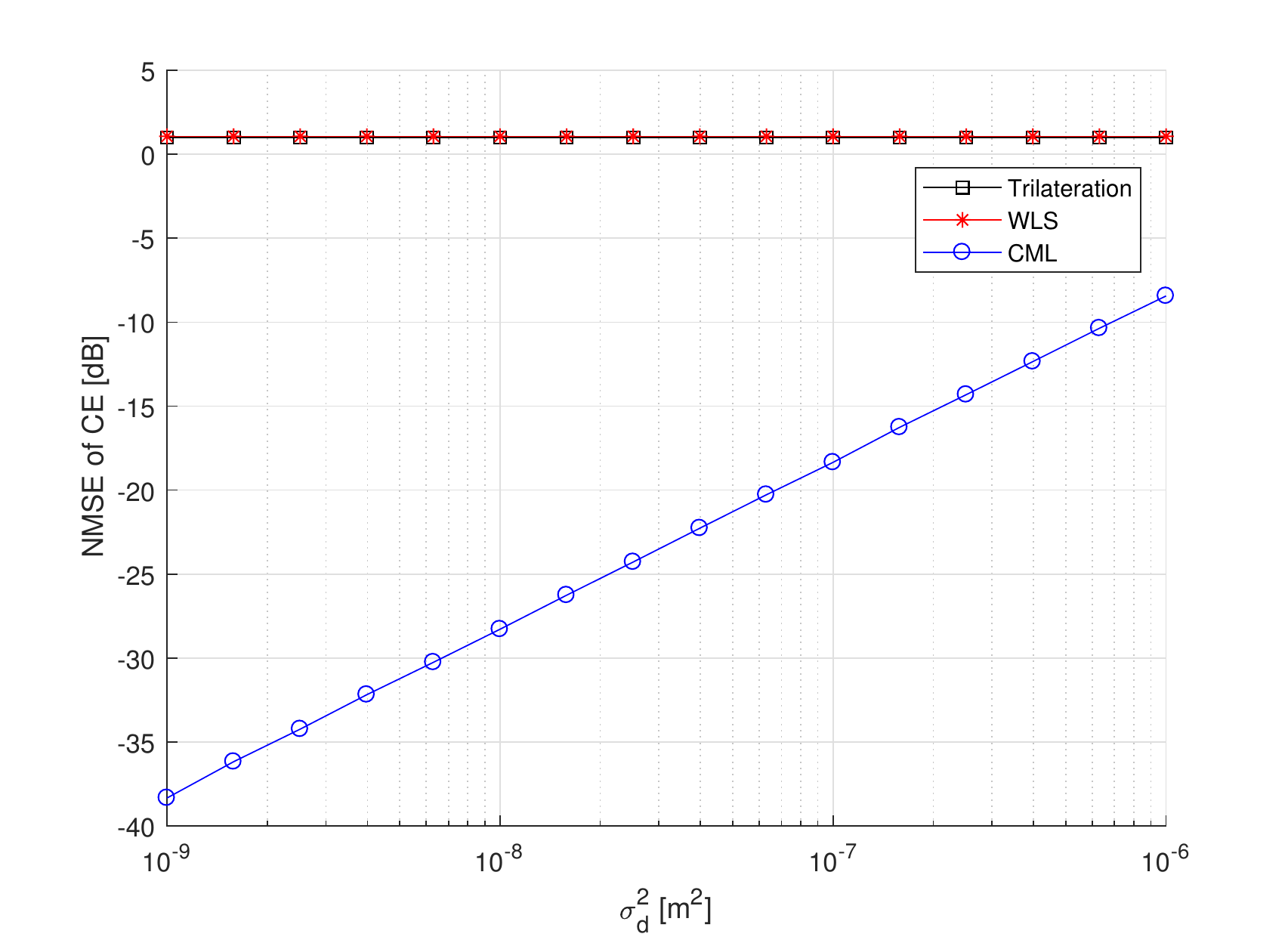}
\caption{The normalized MSE of the channel estimation versus the distance error $\sigma_d^2$.}
\label{fig:nmseg_d}
\end{figure}

Further, we illustrate the performance of the proposed channel estimation scheme against  the WLS algorithm and the conventional trilateration method \cite{DOUKHNITCH:2008Trila3D}. 
The normalized MSE (NMSE) of the channel estimation is defined as
\begin{equation}\label{eq:NMSE CE}
  \text{NMSE} = \mathbb{E}\Big[\frac{|w-\hat{w}|^2}{|w|^2}\Big],
\end{equation}
where $w$ and $\hat{w}$ denote the true cascaded BS-RIS-UE channel and the estimated channel, respectively. 
Fig. \ref{fig:nmseg_d} shows the NMSE of channel estimation based on the aforementioned methods, and indicates that the proposed scheme can achieve lower NMSE than -15 dB when $\sigma_d^2$ is smaller than $10^{-7}$, which outperforms the other two methods.

We let the UE move on the central perpendicular line (CPL) of the RIS. The total transmit power of the BS is fixed to $P_t=30$ dBm.
We combine the JMMSE distance estimation method and CML localization algorithm, which is labeled as ``Proposed Method". 
The curve obtained with the true position of the UE is marked as ``Upper Bound", meaning that perfect CSI is known.
In the ``Uniform Random Case", the phase shift coefficient of each reflecting element is randomly and uniformly distributed, i.e., $\phi_{n}\sim \mathcal{U}(0,2\pi)$. 
The performance of the JBF algorithm \cite{He:2020CascadedChannel} is also shown as a baseline. 
\begin{figure}[!t]
\centering
\includegraphics[width=3.2in]{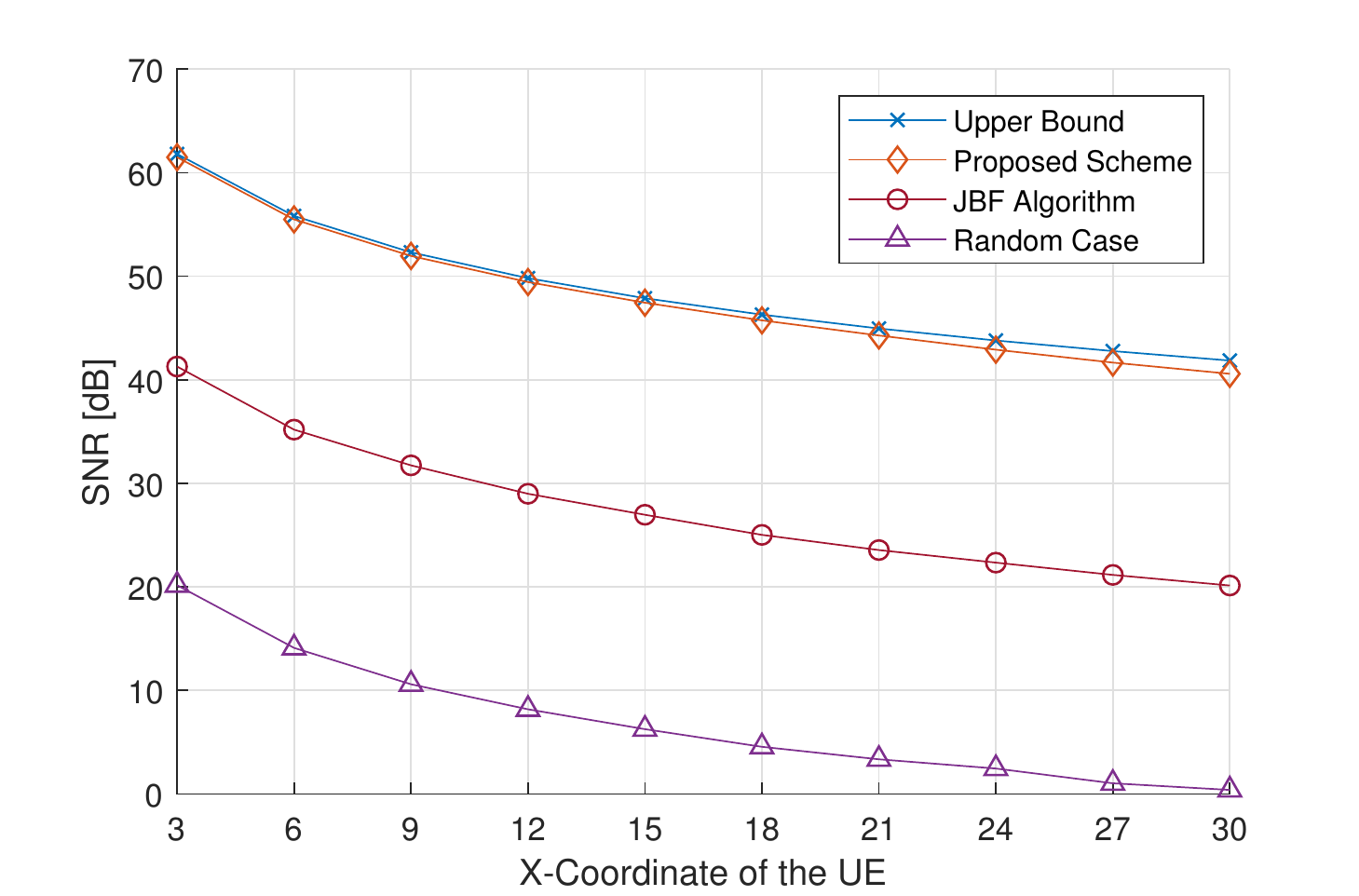}
\caption{The SNR of the received signal via the optimized RIS versus the X-coordinate of the UE.}
\label{fig:snr_x_dB}
\end{figure}

Fig. \ref{fig:snr_x_dB} shows the SNR of the received signal at the UE as a function of the x-coordinate of the UE. The proposed method achieves a  higher SNR than the JBF method, as it approaches the upper bound and has around 40 dB gain over the random phase case.

\begin{figure}[!t]
\centering
\includegraphics[width=3.2in]{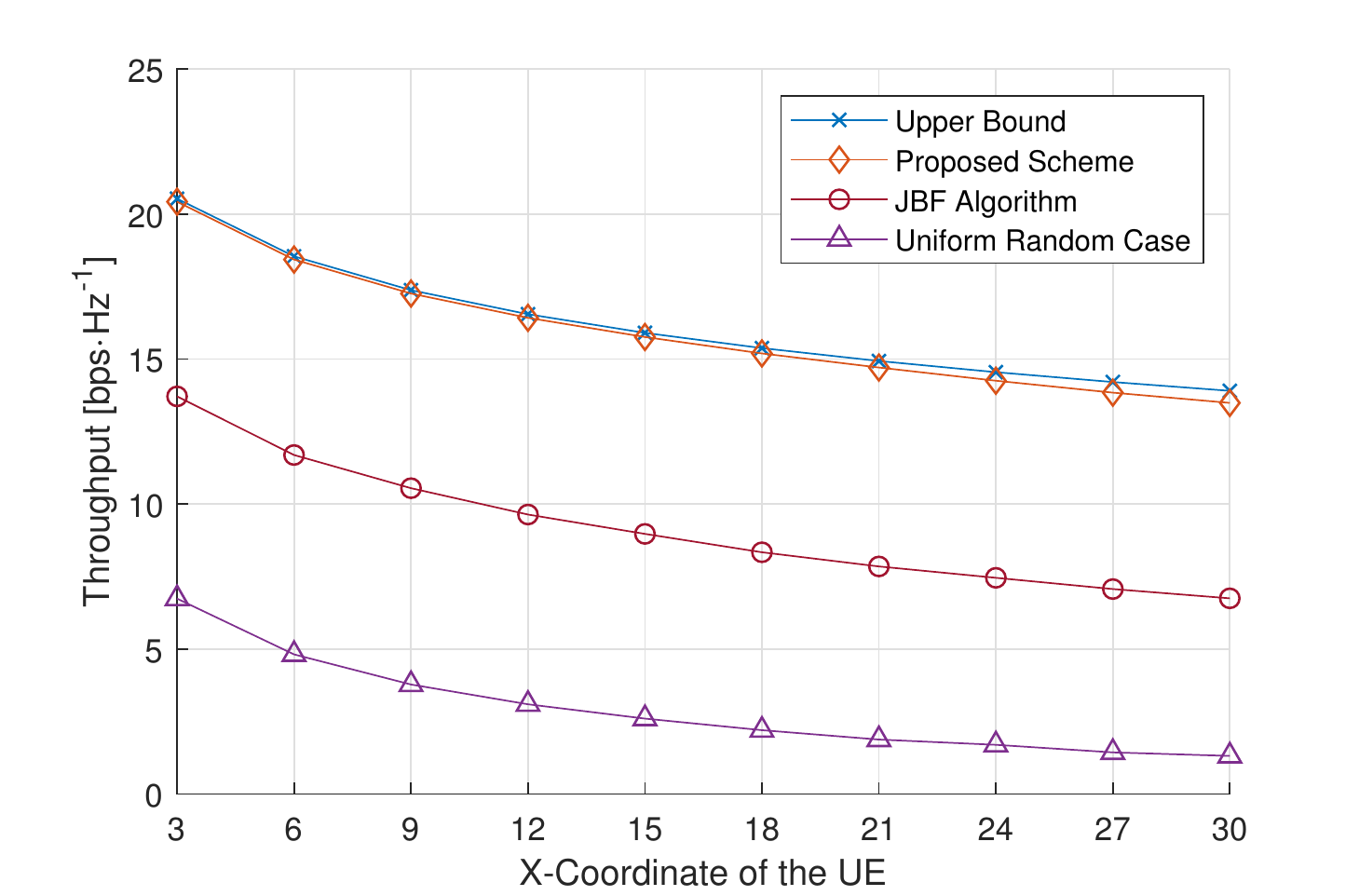}
\caption{The throughput at the receiver versus X-coordinate of the UE.}
\label{fig:th_x}
\end{figure}
Fig. \ref{fig:th_x} shows the throughput versus the x-coordinate of the UE. The throughput is defined as $\log_2(1+\text{SNR})$.
It is observed that the throughput of the proposed method is much higher than the other three schemes and approaches the upper bound. Moreover, the proposed method provides around 13 bps/Hz higher throughput than the random case.

\section{Conclusions}\label{sec:conclusion}
In this paper, we considered the channel estimation of an RIS-aided mmWave communication system. We first proposed the RUS concept for the purpose of positioning the UE based on the structure of the mmWave channel. 
We then proposed a CML 3D localization algorithm with a JMMSE distance estimation method. The CRLBs of the distance estimation error of our proposed method are derived. 
The CSI is reconstructed with the positioning results. 
By exploiting the the sparse multipath angle-delay distribution of the mmWave channel, our method obtains accurate channel estimate we at the cost of very few training resources. Simulation results demonstrated significant performance gain in terms of the SNR over traditional methods.

\appendices
\section{ Proof of Theorem \ref{theo:CRLB} }\label{appendix:CRLB-proof}

The distance between the UE and the $m$-th RUS is defined as:
\begin{equation}\label{eq:dm}
\begin{split}
   d_m & = \|\bm{p}-\bm{p}_m\| \\
     & =[(\bar{x}-x_m)^2+(\bar{y}-y_m)^2+(\bar{z}-z_m)^2]^{\frac{1}{2}},
\end{split}
\end{equation}
the partial derivative of which with respect to $\bm{p}(i),i=1,2,3$, can be calculated as
\begin{equation}\label{eq:dm-x}
  \frac{\partial d_m}{\partial \bm{p}(i)}=\frac{\bm{p}(i)-\bm{p}_m(i)}{d_m}
\end{equation}

The log-likelihood of (\ref{eq:pdf of d est}) is
\begin{equation}\label{eq:P-ln}
  \ln P(\hat{\bm{d}}|\bm{p}) = -\sum_{m=1}^{M}\frac{(\hat{d}_m-d_m)^2}{2\sigma_d^2}-\sum_{m=1}^{M}\ln
  \sqrt{2\pi \sigma_d^2}
\end{equation}
and the partial derivative of (\ref{eq:P-ln}) with respect to $\bm{p}(i)$ is
\begin{equation}\label{eq:ln P}
\begin{split}
   \frac{\partial \ln P(\hat{\bm{d}}|\bm{p})}{\partial \bm{p}(i)} & =-\sum_{m=1}^{M}\frac{\partial}{\partial \bm{p}(i)}\left[ \frac{(\hat{d}_m-d_m)^2}{2\sigma_d^2} \right]\\
     & = \sum_{m=1}^{M}\frac{(\hat{d}_m-d_m)[\bm{p}(i)-\bm{p}_m(i)]}{\sigma_d^2d_m} \\
\end{split}
\end{equation}
Further, the second partial derivative of (\ref{eq:pdf of d est}) with respect to $\bm{p}(i)$ and $\bm{p}(j)$ is
\begin{equation}\label{eq:ln p 2}
\begin{split}
   \frac{\partial^2 \ln P(\hat{\bm{d}}|\bm{p})}{\partial \bm{p}(i)\partial \bm{p}(j)} & = \sum_{m=1}^{M}\frac{1}{\sigma_d^2}\frac{\partial}{\partial\bm{p}(j)} \Big [ \frac{(\hat{d}_m-d_m)[\bm{p}(i)-\bm{p}_m(i)]}{d_m} \Big ] \\
     & = \sum_{m=1}^{M}\frac{1}{\sigma_d^2}\Big[ -\frac{[\bm{p}(i)-\bm{p}_{m}(i)]\cdot[\bm{p}(j)-\bm{p}_{m}(j)]}{d_m^2} \\
     &  + \frac{(\hat{d}_m-d_m)[\bm{p}(i)-\bm{p}_{m}(i)]\cdot[\bm{p}(j)-\bm{p}_{m}(j)]}{d_m^3}\Big]
\end{split}
\end{equation}

Since $\triangle d_m$ follows a Gaussian distribution, $\mathbb{E}[\hat{d}_m-d_m]=0$.
The element of the FIM is expressed as:
\begin{equation}\label{eq:FIM}
\begin{split}
   \bm{\Psi}_{ij} & = -\mathbb{E}\bigg[\frac{\partial^2\ln P(\hat{\bm{d}}|\bm{p})}{\partial\bm{p}(i)\partial\bm{p}(j)}\bigg] \\
     & =\sum_{m=1}^{M}\frac{1}{\sigma_d^2}\frac{[\bm{p}(i)-\bm{p}_{m}(i)]\cdot[\bm{p}(j)-\bm{p}_{m}(j)]}{d_m^2},
\end{split}
\end{equation}

Finally, the CRLB matrix  can be expressed as:
\begin{equation}\label{eq:CRLB}
  \bm{C}=
  \left [\begin{matrix}
    C_{11} & C_{12} & C_{13} \\
    C_{21} & C_{22} & C_{23} \\
    C_{31} & C_{32} & C_{33}
  \end{matrix}\right ]=\bm{\Psi}^{-1}.
\end{equation}

\ifCLASSOPTIONcaptionsoff
  \newpage
\fi

\bibliographystyle{IEEEtran}
\bibliography{allCitations}

\end{document}